\documentclass[11pt,letter]{article}
\usepackage[english]{babel}
\usepackage[margin=0.9in]{geometry}
\usepackage{authblk}
\usepackage{inconsolata}
\usepackage{libertine}
\usepackage{amsmath,amssymb,amsthm,mathrsfs}
\usepackage{mathtools}
\usepackage{mathrsfs,bbold}
\usepackage{xcolor}
\definecolor{ForestGreen}{rgb}{0.1333,0.5451,0.1333}
\definecolor{DarkRed}{rgb}{0.65,0,0}
\definecolor{Red}{rgb}{1,0,0}
\usepackage[linktocpage=true,backref=page,
pagebackref=true,colorlinks,
linkcolor=DarkRed,citecolor=ForestGreen,
bookmarks,bookmarksopen,bookmarksnumbered]
{hyperref}

\usepackage{comment}
\usepackage{mdframed}
\usepackage{xspace}

\usepackage{url}
\usepackage[colorinlistoftodos]{todonotes}

\renewcommand{\leq}{\leqslant}
\renewcommand{\geq}{\geqslant}

\newcommand{\E}{\mathbb{E}}

\newcommand{\Col}{\mathcal{C}}

\newcommand{\eps}{\epsilon}

\usepackage{todonotes}

\newcommand{\Bin}{\textrm{Bin}}

\newcommand{\calC}{\mathcal{C}}

\usepackage{cleveref}

\usepackage{dsfont}

\usepackage{thmtools}
\declaretheorem[numberwithin=section,refname={Theorem,Theorems},Refname={Theorem,Theorems}]{theorem}

\declaretheorem[numberlike=theorem]{lemma}

\declaretheorem[numberlike=theorem]{remark}

\declaretheorem[numberlike=theorem, refname={Observation,Observations},Refname={Observation,Observations},name={Observation}]{observation}

\declaretheorem[refname={Algorithm,Algorithms},Refname={Algorithm,Algorithms},name={Algorithm}]{algorithm}

\title{Simple and Asymptotically Optimal Online Bipartite Edge Coloring}
\author[1]{Joakim Blikstad\thanks{Work done while visiting EPFL.}}
\author[2]{Ola Svensson}
\author[2]{Radu Vintan}
\author[3]{David Wajc\thanks{Work done while at Google Research.}}

\affil[1]{KTH Royal Institute of Technology \&  Max Planck Institute for Informatics}
\affil[2]{EPFL}
\affil[3]{Technion --- Israel Institute of Technology}

\date{\vspace{-1.3cm}}

\begin{document}
\maketitle

\begin{abstract}
    We provide a simple online $\Delta(1+o(1))$-edge-coloring algorithm for bipartite graphs of maximum degree $\Delta=\omega(\log n)$ under adversarial vertex arrivals on one side of the graph. Our algorithm slightly improves the result of (Cohen, Peng and Wajc, FOCS19), which was the first, and currently only, to obtain an asymptotically optimal $\Delta(1+o(1))$ guarantee for an adversarial arrival model.
    More importantly, our algorithm provides a new, simpler approach for tackling online edge coloring.
\end{abstract}

\section{Introduction} \label{sec:Introduction}

Edge coloring is a classic  problem in graph theory and algorithm design: 
\emph{Given a graph, assign colors to the edges, with no two adjacent edges sharing a color}. 
Pioneering work by K\"onig \cite{konig1916graphen} and later Vizing \cite{vizing1964estimate} showed that $\Delta$ and $\Delta+1$ colors suffice for bipartite and general graphs of maximum degree $\Delta$, respectively. (At least $\Delta$ colors are clearly needed.)
Algorithms attaining or approximating these bounds were designed in numerous models of computation, including distributed \cite{panconesi1997randomized,christiansen2023power}, parallel \cite{karloff1987efficient}, dynamic  \cite{duan2019dynamic,christiansen2023power}, and streaming algorithms \cite{charikar2021improved,ansari2022simple,chechik2023streaming,ghosh2023low,behnezhad2023streaming}. The latter includes several simple (asymptotically optimal) $\Delta(1+o(1))$-edge-coloring streaming algorithms for \emph{random-order} streams \cite{charikar2021improved,ansari2022simple}.

In contrast, \emph{online} edge-coloring algorithms (especially for \emph{adversarial order}) and their analyses are somewhat more involved \cite{aggarwal2003switch,bahmani2012online,cohen2019tight,bhattacharya2021online,saberi2021greedy,kulkarni2022online,naor2023online}.
The only truly simple online edge coloring algorithm known is the trivial $2$-approximate greedy algorithm, which is optimal only for the low-degree regime $\Delta=O(\log n)$ \cite{bar1992greedy}.
More involved algorithms were developed for the high-degree setting.
For example, all known algorithms for $(\alpha+o(1))$-approximate adversarial-order online edge coloring with $\alpha<2$ for $\Delta=\omega(\log n)$ \cite{cohen2019tight,saberi2021greedy,kulkarni2022online,naor2023online} rely on interleaved 
invocations of online matching subroutines that compute a matching that matches each edge $e$ with probability at least $1/(\alpha\Delta)$.\footnote{Such matchings can be obtained by sampling a color in an $\alpha\Delta$ coloring, so these problems are basically equivalent.}
The outer loop using such online matching algorithms, introduced by \cite{cohen2019tight}, is not particularly complicated, and can be described and analyzed in about one page (see e.g., \cite[Section 6]{saberi2021greedy}). 
However, the matching algorithms used within this framework and their analyses are quite non-trivial \cite{cohen2018randomized,cohen2019tight,saberi2021greedy,kulkarni2022online,naor2023online}. 

We break from the above template, avoiding this outer loop and subsequent complicated online matching subroutines.
Instead, we obtain our results by a sequence of \emph{offline} bipartite matching computations (more precisely, random sampling of matchings).
This yields a simple asymptotically-optimal online edge coloring algorithm
for the first (and so far only) adversarial arrival model for which positive results are known: one-sided vertex arrivals in bipartite graphs \cite{cohen2019tight}.
Specifically, we prove the following.

  \begin{theorem}[See \Cref{thm:alg}]  There exists an online edge-coloring algorithm  for the one-sided vertex arrival model with the following guarantee. On any $n$-node, maximum degree $\Delta$ bipartite graph, it computes a $(\Delta+q)$-edge-coloring with high probability,\footnote{By \emph{with high probability}, we mean probability of at least $1-n^{-c}$ for some constant $c>0$.} where  $q=O(\Delta^{2/3}\log^{1/3}n)$.
\end{theorem}

The above theorem only gives non-trivial guarantees if $\Delta\geq q$ (i.e., when $\Delta=\Omega(\log n)$ is sufficiently large). Indeed, when $\Delta<q$, greedy already provides an edge coloring with $2\Delta-1\leq \Delta+q$ colors.

Our simple online $\Delta(1+o(1))$-edge-coloring algorithm improves on the $o(1)$ term of the algorithm of \cite{cohen2019tight}, which uses $\Delta+O(\Delta^{3/4}\log^{1/4} n)$ colors if $\Delta=\omega(\log n)$.
Moreover, our simpler algorithm nearly matches a lower bound of $\Delta+\Omega(\sqrt{\Delta})$ colors established in that prior paper. 
We leave the question of whether an algorithm (simple or otherwise) matching this lower bound's $o(1)$ terms exists as an open problem.

\section{Simple yet optimal online bipartite edge coloring}

\paragraph{Problem statement.} A bipartite graph of maximum degree $\Delta=\omega(\log n)$ is revealed.\footnote{As noted above, if $\Delta$ is smaller, the problem is solved optimally by the greedy algorithm.} Initially, only $n, \Delta$ and the nodes on the \emph{offline} side are known. At time $t$, the node $w_t$ on the \emph{online} side is revealed, together with its edges, which must be assigned colors immediately and irrevocably.
The objective is to compute a valid edge coloring using as few colors as possible.

\paragraph{Our algorithm.}
We attempt to provide a valid $(\Delta+q)$-edge-coloring, for $q=o(\Delta)$ to be chosen later.
In particular, we will color edges of each offline node $u$ with distinct colors, chosen uniformly at random from $\Col:=[\Delta+q]$.
To also color edges of each online node $w_t$ with distinct colors, we correlate the random choices at different offline nodes~as~follows.

At each time $t$ we consider a bipartite graph $H_t$ with one side given by the set of neighbors $N_G(w_t)$ of the arriving online node $w_t$ in $G$, and the other side being the set of colors $\Col$. The neighbor $v\in N_G(w_t)$ and color $c\in \Col$ are connected by an edge $cv\in H_t$ if and only if $v$ has no edge colored $c$.\footnote{We use the notation $cv\in H_t$ instead of the more standard but notationally cumbersome $\{c,v\} \in E(H_t)$.}
To color the edges incident to the arriving node $w_t$ in a valid manner, these edges must be given distinct colors and the color chosen for the edge $\{u,w_t\}$ must not already be used at the offline node $u$. These requirements correspond exactly to matchings in $H_t$.
We thus attempt to sample a matching $M_t$ in $H_t$ where each edge $\{u,w_t\}$ is assigned a uniformly random available color of neighbor $u$.
This can be achieved by a number of randomized rounding algorithms for the bipartite matching polytope, provided the desired marginal matching probabilities lie in this polytope. Fittingly, the crux of our analysis is show that the latter holds w.h.p. for $q=o(\Delta)$ sufficiently large.
For simplicity of analysis, we allow for a low-probability ``failure mode'' if this condition fails, in which case we still insist on coloring offline nodes with colors uniformly at random, but without  necessarily providing a valid edge coloring. 
Our pseudocode is given in \Cref{alg:one-sided-arrival}. 

\begin{center}
\begin{minipage}{0.95\textwidth}
\begin{mdframed}[hidealllines=true, backgroundcolor=gray!15]

\begin{algorithm} \ \\[0.2cm]
\emph{At the arrival of online node $w_t$:}
\begin{itemize}
    \item Let $H_t$ be a bipartite graph with node sets $N_G(w_t)$ and $\Col$, with $cv \in H_t$ iff $v$ has no edge colored $c$ (yet).
    \item For each $c \in \Col$ and $v \in V$, let $x^t_{cv} \gets \frac{\mathds{1}[cv \in H_t]}{\Delta-d_t(v)+q}$, for $d_t(v)$ the degree of $v$ by~time~$t$.
    \item If $\sum_v x^t_{cv}\leq 1$ for each color $c\in \Col$: sample matching $M_t$ in $H_t$ with marginals $\Pr[cv\in M_t]=x^t_{cv}$, and color each edge $\{v,w_t\}$ using the color $c$ that is matched to $v$ in $M_t$.
    \item Else (FAILURE MODE): color each edge $\{v,w_t\}$ with u.a.r.~color $c\in N_{H_t}(v)$.
\end{itemize}
\label{alg:one-sided-arrival}
\end{algorithm}
\end{mdframed}
\end{minipage}
\end{center}

\begin{observation}\label{obs:correctness}
By definition, we always have $\sum_{c} x^t_{cv} = 1$ for a vertex $v\in N_G(w_t)$. 
And so, if $\sum_v x^t_{cv}\leq 1$ for all colors $c\in \Col$, then the vector $\vec{x}^t$ is in the bipartite matching polytope (of $H_t$), and  a matching $M_t$ as above can be sampled efficiently (and simply, \cite{gandhi2006dependent}). 
In this case, all edges $\{v,w_t\}$ incident to $w_t$ get colored at time $t$ (since $\sum_{c} x^t_{cv} = 1$) and they all receive distinct colors from their endpoints' prior and other current edges (due to the definition of $H_t$ and $\sum_v x^t_{cv} \leq 1$).
\end{observation}

\paragraph{Analysis overview.} We wish to show that the condition $\sum_v x^t_{cv}\leq 1$ for all times $t$ and colors $c$, necessary to avoid the failure mode and output a valid edge coloring, occurs with high probability.
For this, we prove two invariants in \Cref{lem:invariants}: we prove 
\eqref{eq:uniform_coloring2} a closed form for $\E[x^t_{cv}]$, implying $\E[\sum_v x^t_{cv}] \leq 1-\Omega(q/\Delta)$. If for all $t$ and $c$ these $x^t_{cv}$ were independent, standard Chernoff bounds would suffice to show that w.h.p., $\sum_v x^t_{cv}$ does not deviate much from its expectation, and in particular is at most one. 
As these variables may be dependent, we also prove \eqref{eq:negative correlation2} negative correlation of the random variables $x^t_{cv}$, allowing us to apply Chernoff-like bounds to these dependent variables and prove that the desired condition holds w.h.p., in \Cref{lemma:chernoff}.

\begin{lemma}\label{lem:invariants}
    Let $Z_{cv}^t$ be the indicator variable for color $c$ \emph{not being used} by edges of $v$ when $w_t$ arrives. At any time $t$, the following invariants hold:
\begin{itemize}
    \item \textbf{(Marginals)} For any \emph{color} $c\in \calC$ and \emph{offline} node $v$, we have:
    \begin{equation} \label{eq:uniform_coloring2}
        \Pr[Z_{cv}^t = 1] = \frac{\Delta - d_t(v) + q}{\Delta + q}.
    \end{equation}
    \item \textbf{(Negative dependence)} For any \emph{color} $c\in \Col$ and 
 \emph{offline} nodes $v_1,\dots,v_k$, we have:
    \begin{equation} \label{eq:negative correlation2}
        \Pr\left[\bigwedge_{i\in [k]} (Z^t_{cv_i}=1)\right] \leq \prod_{i \in [k]} \Pr[Z^t_{cv_i}=1].
    \end{equation}
\end{itemize}
\end{lemma}
\begin{proof}
     We prove both invariants by induction on $t\geq 1$. The base case $t=1$ trivially holds for both.
     To prove both inductive steps, we first note that $\mathds{1}[cv \in H_t] = Z^t_{cv} \cdot \mathds{1}[v \in N_G(w^t)]$. So, the value of the random variable $x^t_{cv}$ conditioned on any history up to time $t$ implying  $Z^t_{cv} = 1$ is precisely $\overline{x}^t_{cv} := \frac{\mathds{1}[v \in N_G(w^t)]}{\Delta - d_t(v) + q}$. In particular, conditioning on any such history,  the color $c$ is used for edge $\{v,w_t\}$ with probability precisely $\overline{x}^t_{cv}$ (also in the failure mode, and also if $v\not\in N_G(w_t)$).
     
     The first invariant's inductive step then follows from the above observation and the inductive hypothesis,by a routine calculation, as follows:
     \begin{align}\label{eqn:coloring-step-one-sided}
     \Pr[Z^{t+1}_{cv}=1] & = (1-\overline{x}^t_{cv}) \cdot \Pr[Z^t_{cv}=1] \\
     & = \left(1-\frac{\mathds{1}[v \in N_G(w^t)]}{\Delta - d_t(v) + q}\right)\cdot \frac{\Delta-d_{t}(v)+q}{\Delta+q} \nonumber \\
     & = \frac{\Delta-d_{t+1}(v)+q}{\Delta+q}.\nonumber
     \end{align}
    
     For the second invariant's inductive step, we claim that for any history $\mathcal{H}$ up to time $t$ that implies $\bigwedge_{i\in [k]} (Z^t_{cv_i}=1)$, we have that $\Pr\left[\bigwedge_{i\in [k]} (Z^{t+1}_{cv_i}=1) \; \middle\vert\; \mathcal{H}\right] \leq \prod_{i\in [k]}(1-\overline{x}^t_{cv}).$
     This inequality is clearly an equality for the failure mode, where colors are assigned independently; otherwise, the LHS equals  $1-\sum_{i\in [k]} \overline{x}^t_{cv_i}$, which is upper bounded by the RHS, where this standard inequality follows from the union bound.
     Therefore, by total probability over histories $\mathcal{H}$ as above and the inductive hypothesis and \Cref{eqn:coloring-step-one-sided}, we obtain the claimed statement:
     \begin{align*}
     \Pr\left[\bigwedge_{i\in [k]} (Z^{t+1}_{cv_i}=1)\right] & = \Pr\left[\bigwedge_{i\in [k]} (Z^{t+1}_{cv_i}=1) \;\middle\vert\; \bigwedge_{i\in [k]} (Z^t_{cv_i}=1)\right] \cdot \Pr\left[\bigwedge_{i\in [k]} (Z^t_{cv_i}=1)\right] \\
     & \leq \prod_{i\in [k]}(1-\overline{x}^t_{cv}) \cdot \prod_{i\in [k]} \Pr[Z^t_{cv_i}=1] \\
     & = \prod_{i\in [k]} \Pr[Z^{t+1}_{cv_i}=1].\qedhere
     \end{align*}
\end{proof}

Using these invariants, we now show that \Cref{alg:one-sided-arrival} is unlikely to enter the failure mode.
\begin{lemma} \label{lemma:chernoff}
  If $q=3\Delta^{2/3}\log^{1/3}n\leq \Delta$, then with high probability,
    for each time $t$ and color~$c\in \Col$
    $$\sum_v x^t_{cv}\leq 1.$$
\end{lemma}
\begin{proof}
Fix a time $t$ and color $c$. Notice that $\mathds{1}[cv \in H_t] = Z^t_{cv} \cdot \mathds{1}[v \in N_G(w^t)]$,  and 
hence $x^t_{cv} =  \frac{Z^t_{cv}}{\Delta - d_t(v) + q}$ for all $v\in N_G(w_t)$ (and $x^t_{cv}=0$ for all $v\notin N_G(w_t)$).
For all $v\in N_G(w_t)$, define the random variables $Y_v:=q\cdot x^t_{cv}=\frac{q}{\Delta-d_t(v)+q}\cdot Z^t_{cv}$. It suffices to prove that $\sum_{v} Y_v\leq q$ with high probability.
This follows from a variant of Chernoff bounds, as follows.

First, by Invariant \eqref{eq:negative correlation2}, because $Y_v \neq 0$ if and only if $Z^t_{cv} = 1$, we have that:
\begin{equation*}
    \Pr\left[ \bigwedge_v (Y_v \neq 0) \right] \leq \prod_{v} \Pr\left[ Y_v \neq 0 \right].
\end{equation*}
For such weighted binary variables $Y_v \in \{0, \frac{q}{\Delta-d_t(v)+q}\}$, the above is  equivalent to the definition of $1$-correlation in the sense of \cite[Definition 3.1]{panconesi1997randomized}, namely $\E[\prod_{v\in U} Y_v]\leq \prod_{v\in U}\E[Y_v]$ for all $U\subseteq N_G(w^t)$. As shown in \cite{panconesi1997randomized}, this suffices to upper bound the moment-generating function of $\sum_v Y_v$ and derive strong tail bounds. In particular, by \cite[Corollary 3.3]{panconesi1997randomized}, since we also have that $Y_v \in [0,1]$ for all $v$, the following Chernoff bound holds for any $\varepsilon > 0$:
\begin{equation} \label{eq:chernoff_type_ineq}
    \Pr\left[ \sum_{v} Y_v \geq (1 + \varepsilon) \cdot \E\left[\sum_{v} Y_v\right]\right]  \leq \exp\left( -\frac{\varepsilon^2 \cdot \E\left[\sum_{v} Y_v\right]}{2 + \varepsilon} \right).
\end{equation}

Next, by Invariant \eqref{eq:uniform_coloring2}, $\E[Y_v]=\frac{q}{\Delta+q}$ 
for each node $v\in N_G(w_t)$. Hence, $\E[\sum_{v} Y_v] = \frac{kq}{\Delta + q}$, where $k := |N_G(w_t)| \leq \Delta$.
By setting $\varepsilon := \frac{\Delta + q - k}{k}$ in the Chernoff bound \eqref{eq:chernoff_type_ineq} we obtain:
\begin{align*}
    \Pr \left[ \sum_{v \in N_G(w^t)} Y_v \geq q \right] &= 
    \Pr \left[ \sum_{v \in N_G(w^t)} Y_v \geq \left(1 + \frac{\Delta + q - k}{k} \right) \cdot \frac{kq}{\Delta + q} \right] & \\ 
    &\leq \exp{\left( - \frac{(\Delta + q - k)^2}{k^2} \cdot \frac{kq}{\Delta + q} \cdot \frac{k}{\Delta + q + k} \right) } & \\
    &\leq \exp{ \left(-\frac{q^3}{2\Delta^2 + 3 \Delta q + q^2} \right)} \\
    & \leq \exp{ \left( -\frac{q^3}{6\Delta^2} \right)}.
\end{align*}
Above, the second-to-last inequality follows because $\left( - \frac{(\Delta + q - k)^2 \cdot q}{(\Delta + q)(\Delta + q + k)} \right)$ is decreasing in $k \leq \Delta$, and the last inequality relies on $q\leq \Delta$ by the lemma's hypothesis. 
Thus, for our choice of $q$, $$\Pr\left[\sum_v x^t_{cv}\geq 1\right] = \Pr\left[\sum_v Y_v\geq q\right]\leq \frac{1}{n^{4.5}} \leq \frac{1}{2n^{3}}.$$
The lemma then follows by union bounding over all $n$ online nodes and at most $2n$ colors.
\end{proof}

Combining \Cref{obs:correctness} and \Cref{lemma:chernoff}, we obtain our result.
\begin{theorem}
\label{thm:alg}    \Cref{alg:one-sided-arrival} with $q=3\Delta^{2/3}\log^{1/3}n \leq \Delta$ (i.e., if $\Delta\geq 81\log n$) computes a $(\Delta+q)$-edge-coloring of any $n$-node, maximum degree $\Delta$ bipartite graph with high probability.
\end{theorem}
\begin{proof}
    By \Cref{lemma:chernoff}, the condition $\sum_{v}x^t_{cv}\leq 1$ holds for all time $t$ and colors $c$ with high probability, which by \Cref{obs:correctness} results in a valid edge coloring using $\Delta+q$ colors.
\end{proof}

\begin{remark}
In \Cref{sec:tight-example}, using standard anti-concentration bounds, we also show that our analysis is tight, i.e., that \Cref{alg:one-sided-arrival} indeed requires $\Delta + \Omega(\Delta^{2/3}\log^{1/3}{n})$ colors to work.
\end{remark}

\paragraph{Acknowledgements.} This work was supported by the Swiss National Science Foundation project 200021-184656 ``Randomness in Problem Instances and Randomized Algorithms'' and by the Swiss State Secretariat for Education, Research and Innovation (SERI) under contract number MB22.00054. Joakim Blikstad is partially supported by the Swedish Research Council (Reg. No. 2019-05622) and the Google PhD Fellowship Program. David Wajc is supported by a Taub Family Foundation ``Leader in Science and Technology'' fellowship.

\bibliographystyle{alpha}
\bibliography{abb,ultimate}

\appendix
\section*{APPENDIX}
\section{Tight example for our algorithm}\label{sec:tight-example}

In the following we show that the bound of $(\Delta+O(\Delta^{2/3}\log^{1/3}n))$ colors is tight for \Cref{alg:one-sided-arrival}, for a wide range of $\Delta$ superlogarithmic (and even polynomial) in $n$.

\begin{lemma}
    For any constant $r\geq 3$, there exists an infinite family of instances with $n$ nodes and maximum degree $\Delta = \Theta(n^{1/r})$, on which \Cref{alg:one-sided-arrival} run with $q = \frac{1}{6r^{1/3}} \cdot \Delta^{2/3}\log^{1/3} n$ fails to output a valid edge coloring with constant probability.
\end{lemma}
\begin{proof}
     For all (sufficiently large) integer  $k$, we let $\Delta:=k-1$ and construct an instance graph with $n := \max\{k,\lfloor k^{r-2} \rfloor\}\cdot (\Delta^2 + 1)\leq k^{r-2} \cdot (\Delta^2 + 1)\leq k^r$ many nodes.
     In the instance, $\Delta = k-1 =\Theta(n^{1/r})$ is the maximum degree of any node in the instance.
     We turn to describing this instance.

     A \emph{gadget} consists of an online node $w_t$ connected to $\Delta$ offline neighbors of degree $\Delta-1$ (before $w_t$ arrives), each of these belonging to disjoint subgraphs. Hence, for any color $c$ these offline nodes are neighbors of $c$ in $H_t$ independently. Our instance consists of $\max\{k,\lfloor k^{r-2} \rfloor\}\geq k$ disjoint (hence independent) such gadgets, each having $\Delta^2 + 1$ nodes and therefore totaling $n$ nodes.

    We now fix the gadget corresponding to some $w_t$. Since the $\Delta$ neighbors $v$ of $w_t$ neighbor $c$ independently in $H_t$, each with probability 
    $\frac{\Delta-d_t(v)+q}{\Delta+q}=\frac{q+1}{\Delta+q}$
    (by Invariant \eqref{eq:uniform_coloring2}), the number of neighbors of $c$ in $H_t$ is distributed as $X=\sum_{v\in N_G(w_t)}Z^t_{cv} \sim \textrm{Bin}(\Delta,\frac{q+1}{\Delta+q})$.
    Since all neighbors $v$ of $c$ in $H_t$ have $d_t(v)=\Delta-1$ and hence $x_{cv} = \frac{1}{q+1}\cdot Z^{t}_{cv}$, \Cref{alg:one-sided-arrival} does not enter failure mode if and only if $|X|\leq q+1$.
    We thus wish to lower bound \begin{align}\label{eqn:making-eps-explicit}
    \Pr[X > q+1] = \Pr[X >(1+q/\Delta)\cdot \E[X]] \geq \Pr[X \geq (1+2q/\Delta)\cdot \E[X]].
    \end{align}
    Let $\eps:=2q/\Delta$. 
    By \cite[Lemma 4]{klein2015number}, for $\varepsilon < 1/2$ such that $\varepsilon^2 \cdot \E[X] \geq 3$ (as we shortly verify is the case here), we have the following asymptotic converse of Chernoff's bound:
    \begin{equation} \label{eq:anticoncentration}
        \Pr\left[X \geq (1+\varepsilon) \cdot \E[X]  \right] \geq \exp{ \left( - 9 \varepsilon^2 \cdot \E[X]  \right)}.
    \end{equation}
To see that the required conditions for applying this inequality hold, 
    first notice that $\varepsilon = \frac{2q}{\Delta} = O(\sqrt[3]{\log n/\Delta})$, and so $\varepsilon<1/2$ for sufficiently large $k$ (and hence for sufficietly large $\Delta=\Theta(n^{1/r})>>\log n$).
    On the other hand, we have that for large enough $k$ (and hence $n$):
    \begin{equation*}
        \varepsilon^2 \cdot \E[X]  = \frac{(2q)^2}{\Delta^2}  \cdot \Delta \cdot \frac{q+1}{\Delta + q} \geq \frac{4q^2}{\Delta^2}  \cdot \Delta \cdot \frac{q}{2\Delta}=  \frac{2q^3}{\Delta^2} = \frac{1}{108 r} \cdot \log n \geq 3.
    \end{equation*}
    Similarly, using that $n\leq k^r$, we have:
    \begin{equation}\label{eqn:ub-chernoff-exponent}
        \varepsilon^2 \cdot \E[X] = \frac{(2q)^2}{\Delta^2}  \cdot \Delta \cdot \frac{q+1}{\Delta + q} \leq \frac{4q^2}{\Delta^2}  \cdot \Delta \cdot \frac{2q}{\Delta} = \frac{8q^3}{\Delta^2} = \frac{\log n}{27r} \leq \frac{\log k}{9}.    \end{equation}
    Combining the above, we obtain:
    \begin{equation*}
        \Pr[X > q + 1] \stackrel{\eqref{eqn:making-eps-explicit}}{\geq}\Pr[X\geq (1+\eps)\cdot \E[X]] \stackrel{\eqref{eq:anticoncentration}}{\geq} \exp{ \left( - 9 \varepsilon^2 \cdot \E[X] \right)} \stackrel{\eqref{eqn:ub-chernoff-exponent}}{\geq} \exp\left(-9 \cdot \frac{\log k}{9} \right) = \frac{1}{k}.
    \end{equation*}
    Hence, \Cref{alg:one-sided-arrival} enters failure mode on any fixed gadget with probability at least $\frac{1}{k}$. As the instance consists of $\max\{k,\lfloor k^{r-2} \rfloor\}\geq k$ many independent gadgets, the probability that the algorithm does \emph{not enter} failure mode on \emph{any} of them is upper bounded by a constant,
    $\left( 1 - \frac{1}{k} \right)^{k} \leq 1/e$, or put otherwise $\Pr[\textrm{enter failure mode}]\geq 1-1/e$.

    Now, condition on \Cref{alg:one-sided-arrival} entering failure mode, and fix some time $t$ and color $c\in \Col$ for which  $\sum_{v} x^t_{cv}>1$ (i.e., this is a witness for the algorithm entering failure mode). 
    Then, by the preceding discussion, at least $q+2$ neighbors $v$ of $w_t$ in $G$ are neighbors of $c$ in $H_t$, where they all have have degree $\Delta-d_t(v)+1=q+1$.
    Therefore, by the independent coloring in the failure mode, the probability that the algorithm fails in outputting a valid edge coloring since it assigns $c$ to two or more edges of $w_t$ is at least 
    \begin{align*}
    \Pr[\textrm{fail} \mid \textrm{enter failure mode}]  & \geq \Pr\left[\Bin\left(q+2,\frac{1}{q+1}\right)\geq 2\right] \\
    & = 1-\left(1-\frac{1}{q+1}\right)^{q+2} - \frac{q+2}{q+1}\cdot \left(1-\frac{1}{q+1}\right)^{q+1} \\
    & = 1-\left(\frac{q}{q+1}+\frac{q+2}{q+1}\right)\cdot \left(1-\frac{1}{q+1}\right)^{q+1}\\
    & \geq 1-2/e.
    \end{align*}
    Consequently, \Cref{alg:one-sided-arrival} fails with constant probability, at least $(1-1/e)(1-2/e)$, as claimed.
\end{proof}
\color{black}

\end{document}